\documentclass{article}[10]

\usepackage{fullpage}
\usepackage{longtable}
\usepackage{lscape}
\usepackage{afterpage}
\usepackage{caption}
\usepackage{array}
\usepackage[utf8]{inputenc}
\usepackage{amsfonts}
\usepackage{amsthm}
\usepackage{subcaption}
\usepackage[english]{babel}
\usepackage{bm}
\usepackage[ruled,linesnumbered,vlined]{algorithm2e}
\usepackage{graphicx}
\usepackage{amssymb}
\usepackage{amsfonts}
\usepackage{blindtext}
\usepackage{verbatim}
\usepackage{float}
\usepackage{changepage}
\usepackage{listings}
\usepackage{xcolor}
\usepackage{tcolorbox}
\usepackage{parskip}
\usepackage{scrextend}
\usepackage{comment}
\usepackage{enumerate}
\usepackage{dsfont}
\usepackage{enumitem}
\usepackage{xr}
\usepackage{color}
\usepackage{hyperref}
\usepackage{pgfgantt}
\usepackage{pdflscape}
\usepackage{afterpage}
\usepackage{tikz}
\usepackage{pdflscape}
\usepackage{amsmath}
\usepackage[numbers]{natbib}
\usepackage{hhline}
\usepackage{booktabs}
\usepackage{mathtools}
\usepackage[nodisplayskipstretch]{setspace}
\usepackage{cleveref}

\usetikzlibrary{arrows}

\newcommand{\R}{\mathbb{R}}
\newcommand{\N}{\mathbb{N}}

\renewcommand{\P}{\mathbb{P}}

\newcommand{\1}{\mathds{1}}

\theoremstyle{plain}

\newtheorem{proposition}{Proposition}

\theoremstyle{definition}

\newtheorem{definition}{Definition}[section]

\DeclareMathOperator*{\argmax}{arg\,max}

\newcolumntype{L}[1]{>{\raggedright\let\newline\\\arraybackslash\hspace{0pt}}m{#1}}
\newcolumntype{C}[1]{>{\centering\let\newline\\\arraybackslash\hspace{0pt}}m{#1}}
\newcolumntype{R}[1]{>{\raggedleft\let\newline\\\arraybackslash\hspace{0pt}}m{#1}}

\begin{document}
\title{The Generalized Cascade Click Model: A Unified Framework for Estimating Click Models}



\makeatletter
\renewcommand\@date{{%
  \vspace{-\baselineskip}%
  \large\centering
  \begin{tabular}{@{}c@{}}
    Corn\'{e} de Ruijt \textsuperscript{1} \\
    \normalsize c.a.m.de.ruijt@vu.nl
  \end{tabular}%
  \quad
    \begin{tabular}{@{}c@{}}
    Sandjai Bhulai \textsuperscript{1} \\
    \normalsize s.bhulai@vu.nl
  \end{tabular}

  \bigskip

  \textsuperscript{1}Faculty of Science\\ Vrije Universiteit Amsterdam\\ Amsterdam, the Netherlands \par

  \bigskip

  \today
}}
\makeatother

\maketitle

\begin{abstract}
Given the vital importance of search engines to find digital information, there has been much scientific attention on how users interact with search engines, and how such behavior can be modeled. Many models on user - search engine interaction, which in the literature are known as click models, come in the form of Dynamic Bayesian Networks. Although many authors have used the resemblance between the different click models to derive estimation procedures for these models, in particular in the form of expectation maximization (EM), still this commonly requires considerable work, in particular when it comes to deriving the E-step. What we propose in this paper, is that this derivation is commonly unnecessary: many existing click models can in fact, under certain assumptions, be optimized as they were Input-Output Hidden Markov Models (IO-HMMs), for which the forward-backward equations immediately provide this E-step. To arrive at that conclusion, we will present the Generalized Cascade Model (GCM) and show how this model can be estimated using the IO-HMM EM framework, and provide two examples of how existing click models can be mapped to GCM. Our GCM approach to estimating click models has also been implemented in the \texttt{gecasmo} Python package.
\end{abstract}


\section{Introduction}

In the last decade, a considerable number of models have been proposed to explain and/or predict web user interaction with search engines \cite{chuklin2015click}. Crucial to many of these so-called click models, is to provide an explanation to the \textit{position bias}: web users tend to be more likely to click on items at the top of a Search Engine Result Page (SERP). Many of these click models are probabilistic graphical models (PGMs), in particular in the form of Dynamic Bayesian Networks (DBNs) \cite{murphy2002dynamic, chuklin2015click}. The choice of using DBNs to model user interaction followed from observations in early eye-tracking studies, which suggest that web users on average evaluate the items in the (SERP) sequentially in a top-down fashion \cite{joachims2005accurately}, and the work of \citet{craswell2008experimental}, who showed that a sequential model, the so-called \textit{cascade model}, outperformed other models that try to explain the position bias.

As the parameters of more complex click models do not have a closed-form maximum likelihood estimator per se, estimation is commonly done via expectation maximization (EM). To our knowledge, most authors derived the E-step manually for their proposed click model in a similar fashion as suggested by \cite[Ch. 4]{chuklin2015click}, i.e., by explicitly writing out the expectation. However, as we will argue in this paper, such a derivation is commonly unnecessary. Many click models can be rewritten as an Input-Output Hidden Markov Model (IO-HMM) \cite{bengio1995input}, for which the solution to the EM algorithm is equivalent to the one obtained from deriving the EM algorithm explicitly for the single click model of interest. 

But perhaps more important, if the model can be rewritten as IO-HMM, the IO-HMM directly provides an explicit expression for the E-step, making use of the forward-backward algorithm, thereby making a separately derivation of the E-step for each click model unnecessary. In other words, if we can model a click model as IO-HMMs, we directly obtain an explicit EM estimation procedure for the parameters of the click model, without having to derive further expressions required by the EM algorithm.

This paper contributes to the current click model literature in two ways. First, we introduce a rather general click model, which we name the Generalized Cascade Model (GCM), and we show that the EM algorithm of GCM and the corresponding IO-HMM have the same solution. Furthermore, we show that GCM is a generalization of many commonly used click models, including the User Browser Model (UBM) and the Chapelle-Zang Model (CZM). Second, as an example, we provide explicit mappings from UBM and CZM to GCM, mappings which can be used in the \verb|gecasmo| package to estimate the parameters of these models. The package fully utilizes the solution equivalence between GCM and IO-HMM, and therefore only requires these mappings. I.e., contrary to existing click model estimation software using EM, \verb|gecasmo| does not require programming the E-step or M-step explicitly for each click model.

This paper is structured as follows. Section \ref{sec:relwork} gives an short overview of the click model literature and the methods used to estimate click models. Section \ref{sec:GCMtoIOHMM} presents the GCM and shows how the EM algorithm of IO-HMMs can be used to estimate GCMs, and uses some properties of GCMs to simplify this estimation. Section \ref{sec:package} shows that UBM and CZM can be mapped to GCM, and provides an example of such a mapping for both click models, which can be used in the \verb|gecasmo| package. Lastly, Section \ref{sec:conclusion} gives a brief conclusion and ideas for further research.

\section{Related work}
\label{sec:relwork}
\paragraph{Probabilistic graphical click models}
As already briefly discussed, a considerable number of click models have been proposed in the literature in the form of Dynamic Bayesian Networks (DBNs) \cite{chuklin2015click}, many inspired by the early work on these type of models by \citet{craswell2008experimental}. Since the work of Crasswell et al., many alternative DBNs have been proposed, in particular to cope with the two main limitations of the cascade model: 1) it cannot model more than one click in a SERP, 2) if no item is clicked, it assumes the web user evaluated all items. 

The two perhaps most well known alternatives to the cascade model are the \textit{user browser model} (UBM) \cite{dupret2008user}, and the \textit{Chapelle-Zhang model} (CZM) \cite{chapelle2009dynamic}, which we will discuss in more depth in Section \ref{subsec:clickmodeltoGCM}. From UBM and CZM, many alternative click models have been derived  \cite{chuklin2015click}[Ch. 8]. These include the partially sequential click model (PSCM), and its generalization, the time-aware click model (TACM) \cite{wang2015incorporating,liu2016time}. Both models can be viewed as generalizations of UBM, in which the depth-first search assumption is replaced by a locally unidirectional examination assumption. 

Unfortunately, CZM is in the literature more commonly known as `the dynamic Bayesian network' model. However, as the model is a strict subset of all DBNs, and to avoid confusion in our terminology, we deliberately renamed the model to CZM in this paper. We do like to stress that the authors of the paper speak correctly of `a dynamic Bayesian network'; the name `the dynamic Bayesian network' seems to have followed from other authors referring to the paper. 


\paragraph{Distributed approach to click models}
More recently, recurrent neural networks (RNNs) have grown in popularity as alternatives to DBNs \cite{borisov2016neural, borisov2018click, deng2018ad}, an approach which \citet{borisov2016neural} name the \textit{distributed approach} to click modeling. The work of \citet{borisov2016neural} is in particular interesting, as it shows that these RNNs are able to  outperform both UBM and CZM in terms of click perplexity, though they come at the cost of increased model complexity. 

The latter may be problematic if one is interested in studying the optimal order of items in a SERP under a certain objective function, as was done by \citet{balakrishnan2008optimal}. Given the large number of parameters in these RNNs and their interactions, such analysis would quickly become intractable. Furthermore, although \citet{borisov2016neural} attempt to explain the features learned by the RNN using t-SNE \cite{maaten2008visualizing}, this analysis is post hoc. Hence, when drawing conclusion based on such analysis, one would be at risk of the narrative fallacy. We therefore argue that, even though a distributed approach may lead to better predictions, DBNs are still valuable in explaining web user behavior and optimizing search engines accordingly.

\paragraph{Estimation of PGMs}
Click models which have no closed form maximum likelihood estimate are most commonly estimated using expectation maximization (EM) \cite{chuklin2015click}[Ch. 4]. EM has been implemented for some so-called `basic click models'\cite{chuklin2015click}[Ch. 3], which include the earlier mentioned UBM and CZM, in a number of open source software packages. These include the \verb|ClickModels| \cite{chucklin2020}, and \verb|PyClick|\cite{chucklin2020b} projects.  

Even though we focus on a typical textbook implementation of EM (e.g., \cite{bishop2006pattern}[pp. 438-439]), other authors also consider alterations to EM, or alternatives to EM, in the context of click models. Especially when the number of parameters in the model grows large, alternatives to the classic EM implementation become more interesting to consider. 

In one of the more complex click models, the \textit{task-centric model}, \citet{zhang2011user} make some additional independence assumptions between latent variables in order to be able to use EM, where these assumptions have little impact on the likelihood distribution, and propose alternative updates to speed up convergence. \citet{wang2013content} allow for covariates in the model, while still relying on EM for estimation. They propose a posterior regularized EM algorithm for click models to cope with noisy clicks and mis-ordered item pairs. Apart from EM, \citet{zhu2010novel} propose a click model named the \textit{general click model}, which is solved using \textit{expectation propagation}. \citet{zhang2010learning} a use probit Bayesian inference approach, which can be applied to estimate various click models. Both \cite{zhu2010novel} and \cite{zhang2010learning} allow for covariates in the model.

\section{On the relationship between GCM and IO-HMM}
\label{sec:GCMtoIOHMM}
\subsection{A brief recap of IO-HMM}
We start by a brief recap of the Input-Output Hidden Markov Model (IO-HMM), as introduced by  \citet{bengio1995input}. We will use $\mathcal{P}(X)$ to denote the set of all parents of some node $X$ in a Dynamic Bayesian Network. 
\begin{definition}
An Input-Output Hidden Markov Model (IO-HMM) is a Dynamic Bayesian Network consisting of observed states $\{\mathbf{x}_t, \mathbf{y}_t\}_{t=1,\hdots, T}$, and latent states $\{z_t\}_{t=1,\hdots,T}$. For $t>1$ we have: $\mathcal{P}(z_t)=\{z_{t-1}, \mathbf{x}_{t}\}$, whereas $\mathcal{P}(z_1)=\{\mathbf{x}_{t}\}$. Furthermore, for all $t$ we have: $\mathcal{P}(\mathbf{y}_t)=\{\mathbf{x}_t, z_t\}$, and $\mathcal{P}(\mathbf{x}_t)=\emptyset$. Here $\mathbf{x}_{t}\in\mathbb{R}^{R_1}$,  $z_{t}\in\{1,\hdots,K\}$, and $\mathbf{y}_{t}\in\mathbb{R}^{R_2}$; $R_1, R_2\in\mathbb{N}$. The transition probabilities are given by 
\begin{equation}
    \varphi_{k^{'}k,t}(\Omega):=\P(z_t=k|z_{t-1}=k^{'}; \mathbf{x}_t, \Omega),
\end{equation}
whereas the probability of being in state $k$ at time $t$ given $\mathbf{x}_{1:t}$ is given by
\begin{equation}
\label{eq:defzeta}
    \zeta_{k,t}(\Omega)=\P(z_t=k|\mathbf{x}_{1:t}, \Omega),
\end{equation}
with $\mathbf{x}_{t^{'}:t}=(\mathbf{x}_{t^{'}},\hdots, \mathbf{x}_t)$, $t^{'}<t$, and $\Omega=\{\vartheta_1,\hdots, \vartheta_P\}$, $\vartheta_p\in(0,1)$, being a set of variables. The emission probability for $\mathbf{y}_t$, given current state $z_t$ and covariate vector $\mathbf{x}_t$, is given by some density function $f_y(\mathbf{y}_t|z_t;\mathbf{x}_t,\Omega)$. A graphic representation of the IO-HMM is given by Figure \ref{fig:IOHMM}.
\begin{figure}
\centering
\begin{tikzpicture}[scale=0.75]

\node at (0,5) {$z_{1}$};
\draw (0,5) circle [radius=0.5];

\node at (0,2) {$\mathbf{x}_{1}$};
\filldraw[fill=gray,opacity=0.5] (0,2) circle (0.5);

\node at (2.5,3.5) {$\mathbf{y}_{1}$};
\filldraw[fill=gray,opacity=0.5] (2.5,3.5) circle (0.5);

\draw [->,thick] (0.5,5) -- (2,3.75);
\draw [->,thick] (0.5,2) -- (2,3.25);
\draw [->,thick] (0,2.5) -- (0,4.5);
\draw [->,thick] (0.5,5) -- (3.5,5);

\node at (4,5) {$z_{2}$};
\draw (4,5) circle [radius=0.5];

\node at (4,2) {$\mathbf{x}_{2}$};
\filldraw[fill=gray,opacity=0.5] (4,2) circle (0.5);

\node at (6.5,3.5) {$\mathbf{y}_{2}$};
\filldraw[fill=gray,opacity=0.5] (6.5,3.5) circle (0.5);

\draw [->,thick] (4.5,5) -- (6,3.75);
\draw [->,thick] (4.5,2) -- (6,3.25);
\draw [->,thick] (4,2.5) -- (4,4.5);
\draw [->,thick] (4.5,5) -- (7.5,5);

\filldraw (8,5) circle (0.05);
\filldraw (8.25,5) circle (0.05);
\filldraw (8.5,5) circle (0.05);

\draw [->,thick] (9, 5) -- (9.5, 5);

\node at (10,5) {$z_{T}$};
\draw (10,5) circle [radius=0.5];

\node at (10,2) {$\mathbf{x}_{T}$};
\filldraw[fill=gray,opacity=0.5] (10,2) circle (0.5);

\node at (12.5,3.5) {$\mathbf{y}_{T}$};
\filldraw[fill=gray,opacity=0.5] (12.5,3.5) circle (0.5);

\draw [->,thick] (10.5,5) -- (12,3.75);
\draw [->,thick] (10.5,2) -- (12,3.25);
\draw [->,thick] (10,2.5) -- (10,4.5);















\end{tikzpicture}
\caption{IO-HMM}
\label{fig:IOHMM}
\end{figure}
\end{definition}

Now let us consider IO-HMM under the lens of expectation maximization (EM). Let $i\in\{1,\hdots,n\}$ be the index of some realization of $\{\mathbf{x}_t, \mathbf{y}_t, z_t\}_{t=1,\hdots,T}$. For brevity, we write $\mathbf{z}^{(i)}_{t^{'}:t}=(z_{t^{'}}^{(i)},\hdots, z_{t}^{(i)})$, $t^{'}<t$ (idem for $\mathbf{y}^{(i)}_{t^{'}:t}$ and $\mathbf{x}^{(i)}_{t^{'}:t}$), and $\mathbf{Z}=(\mathbf{z}_{1:T}^{(1)},\hdots, \mathbf{z}_{1:T}^{(n)})$ (idem for $\mathbf{Y}$ and $\mathbf{X}$). We are interested in maximizing the likelihood function
\begin{equation}
\label{eq:bnlikelihood}
\begin{aligned}
    \mathcal{L}(\Omega; \mathbf{X}, \mathbf{Y})&=\mathbb{P}(\mathbf{Y}|\mathbf{X}, \Omega)\\
    &=\sum_{\mathbf{Z}}\mathbb{P}(\mathbf{Y}, \mathbf{Z}|\mathbf{X}, \Omega),
\end{aligned}
\end{equation}
with respect to $\Omega$. Since it is usually not possible to maximize this function directly, EM rather tries to maximize
\begin{equation}
\label{eq:defQ}
Q(\Omega; \hat{\Omega})=\mathbb{E}_{\mathbf{Z}}\left[l_c(\Omega;\mathbf{Z}, \mathbf{Y}, \mathbf{X})|\mathbf{Y}, \mathbf{X}, \hat{\Omega}\right],
\end{equation}
with respect to $\Omega$, keeping some current estimates $\hat{\Omega}$ constant. Here $l_c(\Omega;\mathbf{Z}, \mathbf{Y}, \mathbf{X})$ is the complete data log-likelihood
\begin{equation}
l_c(\Omega;\mathbf{Z}, \mathbf{Y}, \mathbf{X})=\sum_{i=1}^{n}\log\mathbb{P}(\mathbf{y}_{1:T}^{(i)}, \mathbf{z}_{1:T}^{(i)}|\mathbf{x}_{1:T}^{(i)};\Omega).
\end{equation}
Once some (possibly local) optimum $\Omega^{*}$ has been found, the current estimates are updated: $\hat{\Omega}\leftarrow \Omega^{*}$. This process continues until convergence.

In the case of IO-HMM, $Q(\Omega; \hat{\Omega})$ can be written as 
\begin{equation}
\label{eq:IOHMMESS}
Q(\Omega; \hat{\Omega})=\sum_{i=1}^{n}\sum_{t=1}^{T}\sum_{k=1}^{K}\left(\hat{\zeta}_{k,t}^{(i)}\log f_y(\mathbf{y}_t^{(i)}|z_t^{(i)},\mathbf{x}_t^{(i)}; \Omega)+\sum_{k^{'}=1}^{K}\hat{h}_{k^{'}k,t}^{(i)}\log\varphi_{k^{'}k,t}^{(i)}(\Omega)\right),
\end{equation}
with 
\begin{equation}
\hat{h}_{k^{'}k,t}=\mathbb{E}\left[a_{k,t}^{(i)}a_{k^{'},t-1}^{(i)}|\mathbf{x}_{1:T}^{(i)}, \mathbf{y}_{1:T}^{(i)}, \hat{\Omega}\right],
\end{equation}
\begin{equation}
a_{k,t}^{(i)}=\left\lbrace
\begin{array}{ll}
    1 & \text{if } z_{t}^{(i)} = k \\
    0 & \text{otherwise}
\end{array}\right.,
\end{equation}
and $\hat{\zeta}_{k,t}$, being defined as in \eqref{eq:defzeta}, but conditional on current estimates $\hat{\Omega}$.

\subsection{The generalized cascade model (GCM) and its resemblance to IO-HMM}
\label{subsubsec:GCMtoIOHMM}
As mentioned by \citet{chapelle2009dynamic}, many click models show considerable resemblance with hidden Markov models. However, to our knowledge, this resemblance has so far not been made explicit. Neither has previous literature taken advantage of this resemblance. Which is unfortunate, as we do believe there is a considerable advantage to such an approach. Writing out the updates for the parameters $\Omega$ in a specific click model can be a tedious and error prone job, as is somewhat illustrated in \cite[Ch. 4]{chuklin2015click}. Furthermore, this approach only makes partial use of the resemblance between the different click models to quickly find an estimation procedure for $\Omega$. Using the IO-HMM, we are (only) required to define the latent state space, and transition probabilities $\varphi_{k^{'}k,t}^{(i)}(\Omega)$ corresponding to a particular click model, after which we can use the IO-HMM machinery to evaluate $Q(\Omega,\hat{\Omega})$.

To show that we can indeed use IO-HMM to estimate click models, we will first provide a definition for click models having a cascade effect, that is, where the probability of a user clicking an item at position $t$ in a search engine result page (SERP) depends on clicks/skips on items before $t$, but not after $t$. First, we shall introduce some extra notation. We now let $i$ represent a single query session on a search engine. We define a query session as a single realization of the variables $\{\mathbf{x}_t, \bm{\psi}_t, y_t\}_{t=1,\hdots,T}$, with $\bm{\psi}_t=(\psi_{t,1}, \hdots, \psi_{t,P})$, $\psi_{t,p}\in\{0,1\}$, and $\P(\psi_{t,p}=1|\mathcal{P}(\psi_{t,p}); \mathbf{x}_t, \Omega)=\vartheta_{t,p}$. $\mathbf{x}_t$ has the same interpretation as before, while $y_t\in\{0,1\}$, with $y_t=1$ if item at position $t$ is clicked, and $y_t=0$ if the $t$-th item is not clicked (i.e., skipped). We assume there is exactly one configuration of $\bm{\psi_{t}}$ which leads to a positive probability of $y_t=1$. We call this state the `click state', and denote it by $\mathcal{C}_t$. I.e., we allow which state is the click state to depend on $t$.  

Each query session results in a SERP consisting of items $\mathcal{S}_i\subset \mathcal{V}$, with $\mathcal{V}=\{1,\hdots, V\}$ the set of possible items the search engine may return. For simplicity, we will assume each list $\mathcal{S}_i$ has the same number of items $T$, and each item in $\mathcal{S}_i$ occupies an unique position $t\in\{1,\hdots, T\}$. The item in position $t$ in list $\mathcal{S}_i$ is given by the bijective function $r_i(t)$. Since, like in IO-HMM, we assume realizations $i=\{1,\hdots,n\}$ to be conditionally independent given $\mathbf{x}_{1:T}^{(i)}$, we will briefly omit index $i$ in Definition \ref{def:GCM}.

\begin{definition}
\label{def:GCM}
A Generalized Cascade Model (GCM) is a Dynamic Bayesian Network consisting of observed states $\{\mathbf{x}_t, y_t\}_{t=1,\hdots, T}$, and latent states $\{\bm{\psi}_t\}_{t=1,\hdots,T}$, for which for $t>1$: $\mathcal{P}(\bm{\psi}_t)=\{\bm{\psi}_{t-1}, y_{t-1},  \mathbf{x}_{t}\}$; $\mathcal{P}(\bm{\psi}_1)=\{\mathbf{x}_{t}\}$, and for all $t$: $\mathcal{P}(y_t)=\{\mathbf{x}_t, \bm{\psi}_t\}$, and $\mathcal{P}(\mathbf{x}_t)=\emptyset$. 
\end{definition}
\begin{figure}
\centering
\begin{tikzpicture}[scale=0.75]

\node at (0,5) {$\bm{\psi}_{1}$};
\draw (0,5) circle [radius=0.5];

\node at (0,2) {$\mathbf{x}_{1}$};
\filldraw[fill=gray,opacity=0.5] (0,2) circle (0.5);

\node at (2.5,3.5) {$y_{1}$};
\filldraw[fill=gray,opacity=0.5] (2.5,3.5) circle (0.5);

\draw [->,thick] (0.5,5) -- (2,3.75);
\draw [->,thick] (0.5,2) -- (2,3.25);
\draw [->,thick] (0,2.5) -- (0,4.5);
\draw [->,thick] (0.5,5) -- (3.5,5);
\draw [->,thick] (3,3.5) -- (3.75,4.5);

\node at (4,5) {$\bm{\psi}_{2}$};
\draw (4,5) circle [radius=0.5];

\node at (4,2) {$\mathbf{x}_{2}$};
\filldraw[fill=gray,opacity=0.5] (4,2) circle (0.5);

\node at (6.5,3.5) {$y_{2}$};
\filldraw[fill=gray,opacity=0.5] (6.5,3.5) circle (0.5);

\draw [->,thick] (4.5,5) -- (6,3.75);
\draw [->,thick] (4.5,2) -- (6,3.25);
\draw [->,thick] (4,2.5) -- (4,4.5);
\draw [->,thick] (4.5,5) -- (7.5,5);
\draw [->,thick] (7,3.5) -- (8,4.5);

\filldraw (8,5) circle (0.05);
\filldraw (8.25,5) circle (0.05);
\filldraw (8.5,5) circle (0.05);

\draw [->,thick] (9, 5) -- (9.5, 5);
\draw [->,thick] (9, 3.5) -- (9.75, 4.5);

\node at (10,5) {$\bm{\psi}_{T}$};
\draw (10,5) circle [radius=0.5];

\node at (10,2) {$\mathbf{x}_{T}$};
\filldraw[fill=gray,opacity=0.5] (10,2) circle (0.5);

\node at (12.5,3.5) {$y_{T}$};
\filldraw[fill=gray,opacity=0.5] (12.5,3.5) circle (0.5);

\draw [->,thick] (10.5,5) -- (12,3.75);
\draw [->,thick] (10.5,2) -- (12,3.25);
\draw [->,thick] (10,2.5) -- (10,4.5);















\end{tikzpicture}
\caption{GCM}
\label{fig:GCM}
\end{figure}

A graphical representation of GCM is given in Figure \ref{fig:GCM}. The similarity between GCM and IO-HMM should become quickly apparent: in GCM, transition probabilities also depend on $y_{t-1}$, which is not the case for IO-HMM. On the other hand, GCM assumes $y_t$ to be binary, instead of in $\R^{R_2}$, and $\bm{\psi}_t$ is possibly multi-dimensional. However,
 as Proposition \ref{prop:GCMtoIOHMM} shows, for optimization by EM these distinctions do not matter. 
 
\begin{proposition}
\label{prop:GCMtoIOHMM}
$Q_{\text{GCM}}(\Omega,\hat{\Omega})$ can be written in the form of Eq. \eqref{eq:IOHMMESS}.
\end{proposition}
\begin{proof}
We will first convert the multi-dimensional state space into a single-dimensional one, which can easily be done as all we assumed all latent state variables are binary. E.g., we may use the transformation 
\begin{equation}
\label{eq:bincompress}
z_{t}^{(i)}=\text{bin}(\bm{\psi}_t^{(i)})=\sum_{p=1}^{P}2^{p-1}\psi_{t,p}^{(i)}+1,
\end{equation}
to obtain the single-dimensional discrete state space. Note that the $+1$ is strictly not necessary, but simply added to map to some discrete state space counting from 1.

Using Eq. \eqref{eq:defQ} we now have 
\begin{equation}
\label{eq:qgcmeqqiohmm}
\begin{aligned}
Q_{\text{GCM}}(\Omega,\hat{\Omega})&=\mathbb{E}_{\mathbf{Z}|\mathbf{Y}, \mathbf{X}, \hat{\Omega}}\left[l_c^{\text{GCM}}(\Omega;\mathbf{Z}, \mathbf{Y}, \mathbf{X})\right]\\
&=\mathbb{E}_{\mathbf{Z}|\mathbf{Y}, \mathbf{X}, \hat{\Omega}}\left[\sum_{i=1}^{n}\log\mathbb{P}\left(\mathbf{y}_{1:T}^{(i)}, \mathbf{z}_{1:T}^{(i)}|\mathbf{x}_{1:T}^{(i)};\Omega\right)\right]\\
&=\mathbb{E}_{\mathbf{Z}|\mathbf{Y}, \mathbf{X}, \hat{\Omega}}\left[\sum_{i=1}^{n}\sum_{t=1}^{T}\sum_{k=1}^{K}\1_{\{z_t^{(i)}=k\}}\log\mathbb{P}\left(y_{t}^{(i)}| z_{t}^{(i)}=k, \mathbf{x}_{t}^{(i)}, y_{t-1}^{(i)};\Omega\right) \right.\\
& \quad + \left. \sum_{k^{'}=1}^{K}\1_{\{z_t^{(i)}=k, z_{t-1}^{(i)}=k^{'}\}}\log\mathbb{P}\left(z_t^{(i)}=k|z_{t-1}^{(i)}=k^{'},\mathbf{x}_{t}^{(i)}, y_{t-1}^{(i)};\Omega\right)\right]\\
&=\sum_{i=1}^{n}\sum_{t=1}^{T}\sum_{k=1}^{K}\hat{\zeta}_{k,t}^{(i)}\log\P\left(y_t^{(i)}|z_t^{(i)}=k, \mathbf{x}_t^{(i)},y_{t-1}^{(i)};\Omega\right)\\
& \quad +\sum_{k^{'}}^{K}\hat{h}_{k^{'}k,t}^{(i)}\log\mathbb{P}\left(z_t^{(i)}=k|z_{t-1}^{(i)}=k^{'},\mathbf{x}_{t}^{(i)}, y_{t-1}^{(i)};\Omega\right)\\
&=\sum_{i=1}^{n}\sum_{t=1}^{T}\sum_{k=1}^{K}\hat{\zeta}_{k,t}^{(i)}\log\P\left(y_t^{(i)}|z_t^{(i)}=k, \mathbf{x}_t^{(i)'};\Omega\right)\\
& \quad +\sum_{k^{'}}^{K}\hat{h}_{k^{'}k,t}^{(i)}\log\mathbb{P}\left(z_t^{(i)}=k|z_{t-1}^{(i)}=k^{'},\mathbf{x}_{t}^{(i)'};\Omega\right)\\
&=\sum_{i=1}^{n}\sum_{t=1}^{T}\sum_{k=1}^{K}\left(\hat{\zeta}_{k,t}^{(i)}\log f_y\left(y_t^{(i)}|z_t^{(i)}=k,\mathbf{x}_t^{(i)'};\Omega\right)+\sum_{k^{'}=1}^{K}\hat{h}_{k^{'}k,t}^{(i)}\log\tilde{\varphi}_{k^{'}k,t}^{(i)}(\Omega)\right),
\end{aligned}
\end{equation}
with $\mathbf{x}_t^{(i)'}=\left(\mathbf{x}_t^{(i)}, y_{t-1}^{(i)}\right)$.
\end{proof}
Hence, as a result, when our objective is to estimate a GCM using EM, we can simply model it as it were an IO-HMM, but adding the previous click to the input vector. For notational simplicity, we will in the remainder of this paper write $\mathbf{x}_t^{(i)}$ instead of $\mathbf{x}_t^{(i)'}$. We deliberately write $\tilde{\varphi}_{k^{'}k,t}^{(i)}(\Omega)$, and not $\varphi_{k^{'}k,t}^{(i)}(\Omega)$, for reasons that will become apparent in Proposition \ref{prop:noemmision}. 

\section{On the estimation of GCMs using EM for IO-HMM}
\subsection{Notes on the E-step}
As briefly discussed, we believe that modeling GCMs as IO-HMMs has as main advantage that, given transition probabilities $\tilde{\varphi}_{k^{'}k,t}^{(i)}(\Omega)$, the IO-HMM framework directly provides an EM procedure. Although this approach will not reduce time complexity, we believe it does greatly reduces the effort required of finding expressions for the EM-updates. In particular, given Proposition \ref{prop:GCMtoIOHMM}, \citet{bengio1995input} immediately provides the expression required during the E-step. Since GCM is a simplified case of IO-HMM, we can make use of Proposition \ref{prop:noemmision} to somewhat simplify Expression \eqref{eq:qgcmeqqiohmm}. 

\begin{proposition}
\label{prop:noemmision}
There exists transition probabilities $\{\varphi_{k^{'}k,t}^{(i)}(\Omega)\}_{i=1,\hdots,n}^{t=1,\hdots,T}$, with $k^{'},k\in\{1,\hdots,K+1\}$, such that maximizing $Q_{\text{GCM}}(\Omega,\hat{\Omega})$ is equivalent to maximizing
\begin{equation}
\label{eq:QGCMalt}
Q^{'}_{\text{GCM}}(\Omega,\hat{\Omega})=\sum_{i=1}^{n}\sum_{t=1}^{T}\sum_{k=1}^{K+1}\sum_{k^{'}=1}^{K+1}\hat{h}_{k^{'}k,t}^{(i)}\log\varphi_{k^{'}k,t}^{(i)}(\Omega).
\end{equation}
\end{proposition}
\begin{proof}
As one might expect, our objective here is to absorb the emission probability $f_y(\cdot)$ into the new transition probabilities $\varphi_{k^{'}k,t}^{(i)}(\Omega)$. If we define 
\begin{equation}
\vartheta_{t,y}^{(i)} = f_y\left(y_t^{(i)}|z_t^{(i)}=k,\mathbf{x}_t^{(i)'};\Omega\right),
\end{equation}
then a natural thing to do would be to introduce a new latent binary variable $\psi_{t, y}^{(i)}$, with 
\begin{equation}
\P(\psi_{t, y}^{(i)}=1|\bm{\psi}_{t}^{(i)},\mathbf{x}_t^{(i)'};\Omega)=\vartheta_{t,y}^{(i)}.
\end{equation}
However, this would imply $\mathcal{P}(\psi_{t, y}^{(i)}) = \{\bm{\psi}_{t}^{(i)}, \mathbf{x}_t^{(i)}\}$. Hence, the resulting model would not be a GCM, as $\psi_{t, y}^{(i)}$ depends on the current latent state, not the previous latent state. To circumvent this problem, we can include the dependency of the emission on the current state into the definition of the transition probability. Here we make use of the fact that there exists only one click state. I.e., only when $k=\mathcal{C}_t$ does $\vartheta_{t,y}^{(i)}$ affect the transition probability. 

At this point, it is useful to consider the expression that we, according to Eq. \eqref{eq:qgcmeqqiohmm}, would expect for the transition probability. Let $\bm{\psi}_k^{'}=(\psi^{'}_{k,1},\hdots,\psi_{k,P}^{'})$ be the (unique) vector corresponding to $\text{bin}(\bm{\psi}_k^{'})=k$. Under the definition of GCM (Def. \ref{def:GCM}), $\{\psi_{t,1}^{(i)}, \hdots, \psi_{t,P}^{(i)}\}$ are mutually independent given $\mathbf{x}_{t}^{(i)}$ and $\bm{\psi_{t-1}}^{(i)}$, such that we would obtain
\begin{equation}
\label{eq:varphidef1}
\tilde{\varphi}_{k^{'}k,t}^{(i)}(\Omega)=\prod_{p=1}^{P}\left(\vartheta_{t,p}^{(i)}\right)^{\psi_{t,p}^{'}}\left(1-\vartheta_{t,p}^{(i)}\right)^{1-\psi_{t,p}^{'}}.
\end{equation}
Now let $\P(\psi_{t, y}^{(i)}=1|\bm{\psi}_{t-1}^{(i)},\mathbf{x}_t^{(i)'};\Omega)=\vartheta_{t,y}^{(i)}$, then after including the additional variable, the transition probability becomes
\begin{equation}
\varphi_{k^{'}k,t}^{(i)}(\Omega)=\prod_{p=1}^{P}\left(\vartheta_{t,p}^{(i)}\right)^{\psi_{t,p}^{'}}\left(1-\vartheta_{t,p}^{(i)}\right)^{1-\psi_{t,p}^{'}}\left[\left(\vartheta_{t,y}^{(i)}\right)^{\psi_{t,y}^{'}}\left(1-\vartheta_{t,y}^{(i)}\right)^{1-\psi_{t,y}^{'}}\right]^{\1_{\{k=\mathcal{C}_t\}}}.
\end{equation}
Now let $z_{t}^{(i)'}=\text{bin}\left(\bm{\psi}_t^{(i)'}\right)$, with $\bm{\psi}_t^{(i)'}=\left(\psi_{t,1}^{(i)},\hdots,\psi_{t,P}^{(i)}, \psi_{t,y}^{(i)}\right)$ being the state in the new augmented state space, we obtain $f_y(y_t^{(i)}|z_t^{(i)'}=k,\mathbf{x}_t^{(i)};\Omega)=1$. Hence, indeed this new emission probability does not depend on $\Omega$, and the emission can be ignored in the maximization. 

\end{proof}

It should also be noted that in many click models, including CZM and UBM, the emission probability does not depend on $\Omega$ by definition. Hence, many models do not require an additional latent state variable, and for those models the cardinality of the state space could remain the same. Though the analysis of Proposition \ref{prop:noemmision} still holds if we would have $\vartheta_{t,y}^{(i)}=1$, hence we can always augment the state space. For notational simplicity, we will redefine $K$ and $P$ to represent the size of the augmented state space and augmented number of latent variables at $(i,t)$ respectively.

Dropping superscript $(i)$ for a moment, the E-step is given by Proposition \ref{prop:ESTEP}, which are also known as the forward-backward equations.
\begin{proposition}
\label{prop:ESTEP}
Let $\alpha_{k,t}=\P(\mathbf{y}_{1:t}, z_t=k|\mathbf{x}_{1:t};\Omega)$ and $\beta_{k,t}=\P(\mathbf{y}_{t+1:T}|z_t=k, \mathbf{x}_{t+1:T};\Omega)$, then we obtain:
\begin{equation}
\label{eq:alphadef}
\hat{\alpha}_{k,t}=\1_{\{y_t=1|z_t=k\}}\sum_{k^{'}=1}^{K}\hat{\varphi}_{k^{'}k,t}(\mathbf{x}_{t})\hat{\alpha}_{k^{'},t-1},
\end{equation}
\begin{equation}
\label{eq:betadef}
\hat{\beta}_{k,t}= \sum_{k^{'}=1}^{K}\hat{\varphi}_{kk^{'},t}(\mathbf{x}_{t+1})\hat{\beta}_{k^{'},t+1}\1_{\{y_{t+1}=1|z_{t+1}=k^{'}\}},
\end{equation}
and
\begin{equation}
\label{eq:hdef}
\hat{h}_{k^{'}k,t}=\frac{\hat{\beta}_{k,t}\hat{\alpha}_{k^{'},t-1}\hat{\varphi}_{k^{'}k,t}(\mathbf{x}_t)}{\sum_{\ell=1}^{K}\hat{\alpha}_{\ell,T}}\1_{\{y_{t}=1|z_t=k\}},
\end{equation}
with $\hat{\varphi}_{kk^{'},t}(\mathbf{x}_t)$ the estimated transition probability under $\hat{\Omega}$.
\end{proposition}
\begin{proof}
Follows directly from \cite{bengio1995input}.
\end{proof}

\subsection{Notes on the M-step}
One useful property of GCM is that Eq. \eqref{eq:QGCMalt} can be split into multiple optimization problems.
\begin{proposition}
Let $\bm{\theta}_p$ be the weight vector of dimension $m_p$ for parameter $p$, such that  $\vartheta_{t,p}^{(i)}=f_{\text{act}}^{p}(\mathbf{x}_{t,p}^{(i)};\bm{\theta}_{p})=\P(\psi_{t,p}^{(i)}=1|\mathcal{P}(\psi_{t,p}^{(i)}),\mathbf{x}_{t,p}^{(i)};\Omega)$, with $\mathbf{x}_{t,p}^{(i)}$ the covariate vector for parameter $p$, then
\begin{equation}\argmax\limits_{\Omega}Q_{GCM}(\Omega,\hat{\Omega})=\left\lbrace\argmax\limits_{\bm{\theta}_{p}\in\R^{m_{p}}}Q_{GCM}^{'}(\bm{\theta}_{p},\hat{\Omega})\right\rbrace_{p=1,\hdots,P}.
\end{equation}
\end{proposition}
\begin{proof}
Let $\bm{\psi}_k^{'}=(\psi^{'}_{k,1},\hdots,\psi_{k,P}^{'})$ be the (unique) vector corresponding to $\text{bin}(\bm{\psi}_k^{'})=k$. Note that we now assume this represents the state in the augmented state space, i.e., the vector includes $\psi_{k,y}^{'}$, which was defined in Proposition \ref{prop:noemmision}. Furthermore, let $\mathcal{K}(k)$ be original state that we would obtain by omitting $\psi_{k,y}^{'}$ in $\bm{\psi}_k^{'}$.

We define $\mathcal{A}^{(+)}_{t,p,i}$, $\mathcal{A}^{(-)}_{t,p,i}$ as be the set of all transition $(k^{'},k)\in\{1,\hdots,K+1\}\times\{1,\hdots,K+1\}$ having $\vartheta_{t,p}^{(i)}\in(0,1)$, and for which $\psi_{t,p}^{(i)}=1$ and $\psi_{t,p}^{(i)}=0$ respectively. Then
\begin{equation}
\begin{aligned}
Q^{'}_{\text{GCM}}(\Omega,\hat{\Omega})&=\sum_{i=1}^{n}\sum_{t=1}^{T}\sum_{k=1}^{K+1}\sum_{k^{'}=1}^{K+1}\hat{h}_{k^{'}k,t}^{(i)}\log\varphi_{k^{'}k,t}^{(i)}(\Omega)\\
&=\sum_{i=1}^{n}\sum_{t=1}^{T}\sum_{k=1}^{K}\sum_{k^{'}=1}^{K}\hat{h}_{k^{'}k,t}^{(i)}\sum_{\psi\in\{\psi_{t,1}^{(i)},\hdots,\psi_{t,P}^{(i)}\}\setminus \{\psi_{t,y}^{(i)}\}}\log\P(\psi|\mathcal{P}(\psi_{t,p}^{(i)}),\mathbf{x}_{t,p}^{(i)};\Omega)\\
& \quad + \1_{\{\mathcal{K}(k)=\mathcal{C}_{\mathcal{K}(k)}\}}\log\P(\psi_{t,y}^{(i)}=\psi_{k,y}^{'}|\mathcal{P}(\psi_{t,y}^{(i)}),\mathbf{x}_{t,y}^{(i)};\Omega)\\
&=\sum_{p=1}^{P}\sum_{i=1}^{n}\sum_{t=1}^{T}\sum_{(k^{'},k)\in\mathcal{A}^{(+)}_{t,p,i}}\hat{h}_{k^{'}k,t}^{(i)}\log\P(\psi_{t,p}^{(i)}=1|\mathcal{P}(\psi_{t,p}^{(i)}),\mathbf{x}_{t,p}^{(i)};\Omega)\\
&\quad +\sum_{(k^{'},k)\in\mathcal{A}^{(-)}_{t,p,i}}\hat{h}_{k^{'}k,t}^{(i)}\log\P(\psi_{t,p}^{(i)}=0|\mathcal{P}(\psi_{t,p}^{(i)}),\mathbf{x}_{t,p}^{(i)};\Omega)
\end{aligned}
\end{equation}
Hence, optimizing $Q_{\text{GCM}}(\Omega,\hat{\Omega})$ is therefore equivalent to optimizing 


\begin{equation}
\label{eq:QGCMpart}
Q^{p}_{\text{GCM}}(\bm{\theta}_{p},\hat{\Omega})=\sum_{i=1}^{n}\sum_{t=1}^{T}\sum_{(k^{'},k)\in \mathcal{A}^{(+)}_{t,p,i}}\hat{h}_{k^{'}k,t}^{(i)}\log f_{\text{act}}^{p}(\mathbf{x}_{t,p}^{(i)};\bm{\theta}_{p})+\sum_{(k^{'},k)\in \mathcal{A}^{(-)}_{t,p,i}}\hat{h}_{k^{'}k,t}^{(i)}\log(1-f_{\text{act}}^{p}(\mathbf{x}_{t,p}^{(i)};\bm{\theta}_{p}))
\end{equation}
for each $p\in\{1,\hdots,P\}$ separately.
\end{proof}

\begin{proposition}
Let 
\begin{equation}
\hat{h}_{p}^{(+)}=\sum_{i=1}^{n}\sum_{t=1}^{T}\sum_{(k^{'},k)\in \mathcal{A}^{(+)}_{t,p,i}}\hat{h}_{k^{'}k,t}^{(i)},
\end{equation}
and 
\begin{equation}
\hat{h}_{p}^{(-)}=\sum_{i=1}^{n}\sum_{t=1}^{T}\sum_{(k^{'},k)\in \mathcal{A}^{(-)}_{t,p,i}}\hat{h}_{k^{'}k,t}^{(i)}.
\end{equation}
In case $f_{\text{act}}^{p}(\bm{\theta}_{p}; \mathbf{x}_{t,p}^{(i)})=\theta_{p}=\vartheta_p$, $\theta_p\in(0,1)$, the M-step becomes:
\begin{equation}
\label{eq:onecovcase}
\hat{\vartheta}_{p}=\frac{\hat{h}_{p}^{(+)}}{\hat{h}_{p}^{(+)}+\hat{h}_{p}^{(-)}}
\end{equation}
\end{proposition}
\begin{proof}
Follows directly from \cite[pp. 27-30]{chuklin2015click}.
\end{proof}
Hence, by modeling GCM as an IO-HMM, we find exactly the same EM updates as in \cite{chuklin2015click}. However, IO-HMM also directly gives the E-step ( Proposition \ref{prop:ESTEP}), which in \cite{chuklin2015click} still had to be derived for each variable separately. Furthermore, \citet{bengio1995input} also show that if $\mathbf{x}_t=x_t\in\{1,\hdots,L\}$ for some $L\in\N$, then the M-step also has an analytical solution.

Apart from introducing IO-HMM and its estimation using EM, \citet{bengio1995input} also model \\$\P(\psi_{t,p}^{(i)}|\mathcal{P}(\psi_{t,p}^{(i)}),\mathbf{x}_t;\Omega)$ as the output of a neural network. In that case, one could interpret Equation \eqref{eq:QGCMpart} as a weighted log loss function. This may be a `best of both worlds' scenario: as shown by \cite{borisov2016neural}, recurrent neural networks may perform quite well in untangling the manifold obtained from the input vector $\mathbf{x}_t$, which seems in particular relevant when one uses noisy features such as historic click/skips. However, since we use the loss function of a GCM, we can still interpret the results of in terms of the latent variables defined in the proposed click model, such as attraction, satisfaction, or evaluation. I.e., our model remains in that respect interpretable.  

\subsection{Vectorized EM}
\label{subsubsec:implnotates}
\paragraph{Vectorized E-step}
To conclude this section, we will elaborate shortly on how $Q_{\text{GCM}}$ is optimized in the \verb|gecasmo| package. This relates mostly to rewriting the E-step and M-step in vector notation. Since in the following we consider a single query session, we will for the time being drop superscript $i$. We define
\begin{equation}
    A=
\begin{pmatrix}
\alpha_{10} & \hdots & \alpha_{1T}\\
\vdots & \ddots & \vdots\\
\alpha_{K0} & \hdots & \alpha_{KT}
\end{pmatrix}, \quad
    B=
\begin{pmatrix}
\beta_{10} & \hdots & \beta_{1T}\\
\vdots & \ddots & \vdots\\
\beta_{K0} & \hdots & \beta_{KT}
\end{pmatrix},
\end{equation}
and let
\begin{equation}
  H_t=
\begin{pmatrix}
h_{11,t} & \hdots & h_{1K,t}\\
\vdots & \ddots & \vdots\\
h_{K1,t} & \hdots & h_{KK,t}
\end{pmatrix}, \quad
M_t=\begin{pmatrix}
\varphi_{11,t} & \hdots & \varphi_{1K,t}\\
\vdots & \ddots & \vdots\\
\varphi_{K1,t} & \hdots & \varphi_{KK,t}
\end{pmatrix}, \quad
D=\begin{pmatrix}
\delta_{10} & \hdots & \delta_{1T}\\
\vdots & \ddots & \vdots\\
\delta_{K0} & \hdots & \delta_{KT}
\end{pmatrix},
\end{equation}
with $\delta_{kt}=1$ if state $k$ is the click state at time $t$ (0 otherwise). Note that we only have one click state, i.e., $\sum_{k=1}^{K}\delta_{kt}=1$ for all $t\in\{0,\hdots,T\}$. Let $\odot$ be the element-wise product. Using Expressions \eqref{eq:alphadef}, \eqref{eq:betadef}, and \eqref{eq:hdef}, we obtain the following vectorized E-step (Alg. \ref{alg:estep}).
\begin{algorithm}[htbp]
    $B_{1:K,T} \leftarrow 1$\\
    $A_{1:K, 0} \leftarrow D_{1:K,0}$\\
    \For{$t\leftarrow 1$ \KwTo $T$}{
        $A_{1:K,t} \gets y_{t}\left[ D_{1:K,t} \odot \left(M_t^{\top} A_{1:K,t-1}\right)\right] + \left(1-y_{t}\right)\left[\left(1-D_{1:K,t}\right)\odot \left(M_{t}^{\top} A_{1:K,t-1}\right)\right]$\\
    }
    \For{$t\leftarrow T$ \KwTo $1$}{
        $\Lambda_{t} \gets y_t\left(D_{1:K,t}\odot B_{1:K,t}\right) + (1-y_t)\left[(1- D_{1:K,t})\odot B_{1:K,t}\right]$\\
        $B_{1:K, t-1}\gets M_t\Lambda_{t}^{\top}$\\
        $H_{t}\gets \left[\frac{\Lambda_{t}}{\sum_{k=1}^{K}\alpha_{k,T}} A_{1:K,t-1}^{\top}\right] \odot M_t^{\top}$
    }
 \caption{Vectorized E-step}
 \label{alg:estep}
\end{algorithm}
 
\paragraph{Vectorized M-step}
For convenience, we will slightly rewrite \eqref{eq:QGCMpart}. Let
\begin{equation}
\label{eq:paramactfunction}
c_{k^{'}k,t,p,i}^{(o)}=\left\lbrace
\begin{array}{cc}
     1 & \text{if } (k^{'},k)\in\mathcal{A}_{t,p,i}^{(o)}, \quad o=+ \\
    -1 & \text{if } (k^{'},k)\in\mathcal{A}_{t,p,i}^{(o)}, \quad o=- \\ 
    0 & \text{otherwise}\\
\end{array}\right.,
\end{equation}
\begin{equation}
\label{eq:paramactmat}
I_{t,p,i}^{(o)}=
\begin{pmatrix}
c_{11,t,p,i}^{(o)} & c_{12,t,p,i}^{(o)} & \cdots & c_{1K,t,p,i}^{(o)} \\
\vdots & \vdots & \ddots &\vdots\\
c_{K1,t,p,i}^{(o)} & c_{K2,t,p,i}^{(o)} & \cdots & c_{KK,t,p,i}^{(o)} \\
\end{pmatrix},
\end{equation}
and
\begin{equation}
\label{eq:weightcomp}
w_{t,p,i}^{(o)}= \left(\left(\hat{H}_t^{(i)} \odot I_{t,p,i}^{(o)}\right)\mathbf{1}\right)^{\top}\mathbf{1},
\end{equation}
with $\mathbf{1}$ a vector of ones of size $K$, $o\in\{+,-\}$, and $\odot$ being the element-wise multiplication operator. Let $w_{t,p,i} = w_{t,p,i}^{(+)}+w_{t,p,i}^{(-)}$, now \eqref{eq:QGCMpart} can be rewritten as
\begin{equation}
\begin{aligned}
\label{eq:QGCMimp}
Q^{p}_{\text{GCM}}(\bm{\theta}_{p},\hat{\Omega})&=\sum_{i=1}^{n}\sum_{t=1}^{T}w_{t,p,i}^{(+)}\log f_{\text{act}}^{p}(\mathbf{x}_t;\mathbf{\bm{\theta}}_{p})+ (-w_{t,p,i}^{(-)})\log\left(1-f_{\text{act}}^{p}(\mathbf{x}_t;\mathbf{\bm{\theta}}_{p})\right)\\
&=\sum_{i=1}^{n}\sum_{t=1}^{T}|w_{t,p,i}|\left[\frac{\text{sgn}(w_{t,p,i})+1}{2}\log f_{\text{act}}^{p}(\mathbf{x}_t;\mathbf{\bm{\theta}}_{p})\right.\\
&\left.\quad+\left(1-\frac{\text{sgn}(w_{t,p,i})+1}{2}\right)\log(1-f_{\text{act}}^{p}(\mathbf{x}_t;\mathbf{\bm{\theta}}_{p}))\right].
\end{aligned}
\end{equation}
In the latter expression, we conveniently make use of the fact that for some triple ($t,p,i$), either $w_{t,p,i}^{(+)}$ or $w_{t,p,i}^{(-)}$ is non-zero. Hence, we only have to keep one $T\times P \times n$ tensor in memory, and use the sign of the weight to determine whether the weight is with respect to $\psi_{t,p}^{(i)}=1$ or $\psi_{t,p}^{(i)}=0$. Although the former expression in \eqref{eq:QGCMimp} has the same form as binary cross-entropy, we have $w_{t,p,i}^{(+)}+(-w_{t,p,i}^{(-)})\leq 1$, hence it should be interpreted as a weighted log-loss function. 

Finally, the GCM's EM algorithm is summarized in Algorithm \ref{alg:gcmem}.
\begin{algorithm}[htbp]
  \SetKwInOut{Input}{inputs}
  \SetKwInOut{Output}{output}
      \Input{Initial parameter estimates $\hat{\Omega}$;\\ 
      Covariates matrices $\{X_{p}\}_{p=1,\hdots,P}$;\\
      Parameter activation functions $\{f_{\text{act}}^{p}(\cdot)\}_{p=1,\hdots,P}$; \\ 
      Parameter activation matrices $\{I_{p,t}^{(+)}, I_{p,t}^{(-)}\}_{p=1,\hdots,P}^{t=t\hdots,T}$; \\
      Transition matrices $\{M_t^{(i)}\}_{i=1,\hdots,n}^{t=1,\hdots,T}$; \\
      Tolerance parameter $\epsilon$}; 
    \Output{Fitted parameters $\{\hat{\bm{\theta}}_{p}\}_{p=1,\hdots,P}$}
    Choose some initial parameters $\{\hat{\bm{\theta}}_p\}_{p=1,\hdots,P}$, and $\Delta>\epsilon$;\\ 
    \While{$\Delta>\epsilon$}{
    Run the E-step (Alg. \ref{alg:estep}), given current estimates $\hat{\Omega}$;\\ 
    Compute the weights for each variable $p$ using Expression \eqref{eq:weightcomp};\\
    Find for each variable $p$ new weights $\hat{\bm{\theta}}_{p}^{'}$ by optimizing \eqref{eq:QGCMimp};\\
    $\Delta = \sum_{p=1}^{p}|\hat{\bm{\theta}}_p-\hat{\bm{\theta}}_p^{'}|$;\\
    Set $\hat{\bm{\theta}}_{p}\leftarrow \hat{\bm{\theta}}_{p}^{'}$;\\
    }
     \Return $\{\hat{\bm{\theta}}_{p}\}_{p=1,\hdots,P}$;
 \caption{GCM's EM procedure}
 \label{alg:gcmem}
\end{algorithm}

\section{Modeling click models as GCM}
\label{sec:package}
\subsection{Mapping click models to GCM}
\label{subsec:clickmodeltoGCM}
In Section \ref{subsubsec:GCMtoIOHMM}, we introduced the generalized cascade model, and showed that optimization with EM is equivalent to that of the IO-HMM model, given that we add the click/skip at position $t-1$ to the input vector $\mathbf{x}_t^{(i)}$. In this section, we will show that common click models such as the Chapelle-Zhang model (CZM) \cite{chapelle2009dynamic}, and the User Browser Model (UBM) \cite{dupret2008user}, can be rewritten as GCMs. As a direct consequence, any simplification of CZM or UBM (such as the cascade model \cite{craswell2008experimental} or the dependent click model \cite{guo2009efficient}) are also GCMs. In fact, any click model where the latent variables can be combined into one discrete latent state variable $z_t$, such that the remaining model is still Markovian w.r.t. its parent as defined in GCM's definition (Def. \ref{def:GCM}), can be modeled as GCM. Hence, also more general models than CZM and UBM can be modeled as GCM.

\begin{definition}
\label{def:czmmodeldef}
The Chapelle-Zhang model (CZM) \cite{chapelle2009dynamic} is a Dynamic Bayesian Network consisting of latent binary variables $\{R_{r_i(t)}^{(i)}, S_{r_i(t)}^{(i)}, E_{t}^{(i)}\}^{i=1,\hdots,n}_{t=1,\hdots,T}$, and observed binary variables $\{y_t^{(i)}\}^{i=1,\hdots,n}_{t=1,\hdots,T}$, where $R_{r_i(t)}^{(i)}=1$ if the user finds the item at position $r_i(t)$ attractive (0 otherwise), $S_{r_i(t)}^{(i)}=1$ if the user is satisfied with the item at position $r_i(t)$ after having clicked the item (0 otherwise), $E_{t}^{(i)}=1$ if the item at position $t$ is evaluated by the user (0 otherwise), and $y_t^{(i)}=1$ if the item at position $t$ is clicked (0 otherwise). Given SERP $\mathcal{S}_i$, users navigate through the SERP according to  
\begin{equation}
\label{eq:SDBNdef3}
    \P(R_{r_i(t)}^{(i)}=1)=\phi_{r_i(t)}^{(R)},
\end{equation}
\begin{equation}
\label{eq:SDBNdef4}
    \P(S_{r_i(t)}^{(i)}=1|y_t^{(i)}=c)=\left\lbrace
    \begin{array}{ll}
        \phi_{r_i(t)}^{(S)} & \text{if } c=1 \\
        0 & \text{if } c=0
    \end{array}\right.;
\end{equation}
\begin{equation}
\label{eq:SDBNdef5}
\P(E_{t}^{(i)}=1| E_{t-1}^{(i)}=a, S_{r_i(t-1)}^{(i)}=b)=\left\lbrace
\begin{array}{ll}
    \gamma & \text{if } a=1, b=0 \\
    0 & \text{otherwise} 
\end{array}\right.,\quad t>1;
\end{equation}
\begin{equation}
\label{eq:SDBNdef6}
E_{1}^{(i)}=1;
\end{equation}
\begin{equation}
\label{eq:SDBNdef9}
    y_t^{(i)}=1 \iff E_{t}^{(i)}=1, R_{r_i(t)}^{(i)}=1.
\end{equation}
\end{definition}

\begin{definition}
\label{def:UBM}
The User Browser Model (UBM) \cite{dupret2008user} is a Dynamic Bayesian Network consisting of latent binary variables $\{R_{r_i(t)}^{(i)}, E_{t}^{(i)}\}^{i=1,\hdots,n}_{t=1,\hdots,T}$, and observed variables $\{y_t^{(i)}\}^{i=1,\hdots,n}_{t=1,\hdots,T}$, all having the same interpretation as in CZM, for which we assume users navigate through a SERP according to
\begin{equation}
\label{eq:ubmdef1}
    \P(R_{r_i(t)}^{(i)}=1)=\phi_{r_i(t)}^{(R)},
\end{equation}
\begin{equation}
\label{eq:ubmdef2}
    \P(E_{t}^{(i)}=1|y_{t^{'}}^{(i)}=1)=\gamma_{t^{'}t}, \quad t^{'}<t
\end{equation}
\begin{equation}
\label{eq:ubmdef3}
y_{0}^{(i)}=1;
\end{equation}
\begin{equation}
\label{eq:ubmdef5}
    y_{t}^{(i)}=1 \iff E_{t}^{(i)}=1, R_{r_i(t)}^{(i)}=1.
\end{equation}
\end{definition}

In Propositions \ref{prop:czmisgcm} and \ref{prop:ubmisgcm} again we briefly drop superscript $i$.

\begin{proposition}
\label{prop:czmisgcm}
CZM is a GCM
\end{proposition}
\begin{proof}
 For simplicity, we again drop subscript $i$. Let $\bm{\psi}_t=(S_{t-1}, A_t, E_t)$ and $z_t=\text{bin}(\bm{\psi}_t)$.  For completeness, we introduce auxiliary variables $S_0=0$ and $x_t=1$ for all $t$. Then indeed we find $\mathcal{P}(z_t)=\{z_{t-1}, y_{t-1},x_t\}$ for $t>1$,  $\mathcal{P}(z_1)=\{x_1\}$, $\mathcal{P}(y_t)=\{x_t, z_t\}$, and $\mathcal{P}(x_t)=\emptyset$ for all $t$. 
\end{proof}

Hence, the idea of modeling click models as GCM is to combine the binary latent variables $\psi_{t,1}^{(i)},\hdots, \psi_{t,P}^{(i)}$ between two subsequent clicks into one latent variable by using Equation \eqref{eq:bincompress}. I.e., the size of the state space of $z_t$ becomes $\mathcal{O}(2^P)$. If after this transformation the click model is Markovian with respect to a subset of the parents as defined in GCM, then we can conclude that the model is a GCM, and we are allowed to use the IO-HMM machinery to estimate the model's parameters using EM. 

Perhaps less obvious is that UBM can also be modeled as GCM. Looking at the definition of UBM (Definition \ref{def:UBM}), we find that UBM is Markovian in the last clicked item, hence, not Markovian in the previous position as is required by GCM, something which is called the \textit{first-order examination assumption} in click literature \cite{liu2016time}. However, we can pass information about the last clicked item from position $t$ to position $t+1$, by expanding the state space, thereby making UBM Markovian in the previous position.
\begin{proposition}
\label{prop:ubmisgcm}
UBM is a GCM
\end{proposition}
\begin{proof}
From the definition of UBM, we find that it only has, at some time $t$, two latent variables: $E_t$ and $A_t$. The parents are given by $\mathcal{P}(A_t)=\emptyset$, and $\mathcal{P}(E_t)=\{y_{t^{'}}\}$, with $t^{'}=\max\{\tilde{t}\in\{0,\hdots,t-1\}|y_{\tilde{t}}=1\}$. We now define $\mathbf{e}_t=(e_{1,t}, \hdots, e_{T,t})$, where
\begin{equation}
    e_{t^{'},t}= \left\lbrace
    \begin{array}{ll}
        \1_{\{t^{'}<t\}}e_{t^{'}, t-1}(1-y_{t}) + \1_{\{t^{'}=t\}}y_{t} & \text{if } t>0  \\
        1 & \text{if } t=0
    \end{array}\right..
\end{equation}
I.e., $\mathbf{e}_{t}$ stores the last clicked item before position $t$. Since $e_{t^{'}, t}\in\{0,1\}$, we can define $\bm{\psi}_t=(\mathbf{e}_t,A_t, E_t)$ and $z_t=\text{bin}(\bm{\psi}_t)$ to again obtain a one-dimensional discrete latent state space. Let $x_t=1$ for all $t$, then indeed we find $\mathcal{P}(z_t)=\{z_{t-1}, y_{t-1},x_t\}$ for $t>1$; $\mathcal{P}(z_1)=\{x_1\}$, $\mathcal{P}(y_t)=\{x_t, z_t\}$, and $\mathcal{P}(x_t)=\emptyset$ for all $t$.
\end{proof}

Following the derivations of Propositions \ref{prop:czmisgcm} and \ref{prop:ubmisgcm}, we also notice that we could for example replace $\phi^{(A)}_{r_i(t)}$, $\phi^{(S)}_{r_i(t)}$, and $\gamma$ with activation functions $f_{\text{act}}^{(A)}(\mathbf{x}_t^{(A)}; \bm{\theta}^{(A)})$, $f_{\text{act}}^{(S)}(\mathbf{x}_t^{(S)}; \bm{\theta}^{(S)})$, and $f_{\text{act}}^{(E)}(\mathbf{x}_t^{(E)}; \bm{\theta}^{(E)})$ in case of CZM. To ensure the output of these activation functions are probabilities, we ensure the activation functions all map to the interval $(0,1)$. Here $\bm{\theta}=(\bm{\theta}^{(A)}, \bm{\theta}^{(S)}, \bm{\theta}^{(E)})$ is a vector of unknown parameters, and $\mathbf{x}_t^{(A)}, \mathbf{x}_t^{(S)}, \mathbf{x}_t^{(E)}$ is observed data, stored in $\mathbf{x}_t$, i.e., $\mathbf{x}_t=(\mathbf{x}_t^{(A)}, \mathbf{x}_t^{(S)}, \mathbf{x}_t^{(E)})$. The same argument holds for UBM. As a consequence, also generalizations of CZM and UBM can be modeled as GCM (e.g., the Bayesian Sequential State Model \cite{wang2013content}, or many of the `advanced click models' described in \citet{chuklin2015click}[Ch. 8]), as long as the Markovian properties of GCM are met. 

In terms of the cardinality of the state space, we do notice that adding the vector of the last clicked item to the state space, as we have done with UBM, at first seems to lead to a considerable memory burden. Where CZM has a cardinality smaller than $2^{3}+1$ (as some combinations of $\{R_t,E_t,S_{t-1}\}$ are infeasible), for UBM this becomes $\mathcal{O}(2^{T})$. However, it should be noted that in most cases the number of positions $T$ can be taken small. Click probabilities have a geometric decay in $T$, and many empirical studies show that these probabilities are already commonly already negligible for $T>10$ (e.g., \cite{craswell2008experimental, borisov2016neural}). Hence, such state spaces are not likely to become a burden in terms of computational complexity or storage. 

\subsection{Explicit mappings}
In the following examples, we will provide an explicit mapping from CZM and UBM to GCM. That is, how the transition matrix and activation matrices of these two models could be defined to be able to run the GCM's EM procedure (Alg. \ref{alg:gcmem}). These example mappings have also been implemented and can by found on the \verb|gecasmo|'s Github page \cite{deruijt2020a}  



\paragraph{CZM to GCM}
\label{subsubsec:explic_czm_gcm}
We will first consider how CZM can be written as GCM. As Proposition \ref{prop:czmisgcm} already showed, CZM can be modeled as a GCM by defining the latent state as $\bm{\psi}_t^{(i)}=(S_{r_i(t-1)}^{(i)}, R_{r_i(t)}^{(i)}, E_t^{(i)})$, where $S_{r_i(t-1)}^{(i)}$, $R_{r_i(t)}^{(i)}$, and  $E_t^{(i)}$ are interpreted as in Definition \ref{def:czmmodeldef}. One of the possible discretizations of the state space is given by Table \ref{tab:czmspmapping}. We have 3 binary variables, which normally would lead to a state space of size 8. However, we add an additional absorbing state $O$, and the two states having $(S_{r_i(t-1)}^{(i)}=1, E_t^{(i)}=1)$ are infeasible, therefore we end up with a total of 7 states. The  click state represents the state in which an item is clicked, and as discussed before, we assume there is only one such state for each $t$.  For convenience, we also added the corresponding click value $y_t^{(i)}$, though it is not part of the state space.

\begin{table}[htbp]
    \centering
    \begin{tabular}{l|ccc|r}
    \toprule
        State & $S_{r_i(t-1)}^{(i)}$ & $E_t^{(i)}$ & $R_{r_i(t)}^{(i)}$ & $y_t^{(i)}$ \\
        \hline
        $0$ & $0$ & $0$ & $0$ & $0$  \\  
        $1$ & $0$ & $0$ & $1$ & $0$  \\  
        $2$ & $0$ & $1$ & $0$ & $0$  \\  
        $3$ $(\mathcal{C}_t)$ & $0$ & $1$ & $1$ & $1$  \\  
        $4$ & $1$ & $0$ & $0$ & $0$  \\  
        $5$ & $1$ & $0$ & $1$ & $0$  \\
        $6$ $(O)$ & - & - & - & - \\
        \hline
        - &$\textit{1}$ & $\textit{1}$ & $\textit{0}$ & $0$  \\  
        - & $\textit{1}$ & $\textit{1}$ & $\textit{1}$ & $1$  \\  
        \bottomrule
    \end{tabular}
    \caption{CZM state space mapping}
    \label{tab:czmspmapping}
\end{table}

Using this state mapping and the probabilities introduced in Definition \ref{def:czmmodeldef}, and by writing $\bar{x}=1-x$, we find the transition matrix.
\begin{equation}
\label{eq:czmtransmat}
M_{i,t}=\begin{pmatrix}
0 & 0 & 0 & 0 & 0 & 0 & 1 \\
0 & 0 & 0 & 0 & 0 & 0 & 1 \\
\bar{\gamma_i}\bar{\phi}^{(A)}_{r_{i}(t)} & \bar{\gamma_i}\phi^{(A)}_{r_{i}(t)} & \gamma_i\bar{\phi}^{(A)}_{r_{i}(t)} & \gamma_i\phi^{(A)}_{r_{i}(t)} & 0 & 0 & 0\\
\bar{\gamma_i}\bar{\phi}^{(A)}_{r_{i}(t)}\bar{\phi}^{(S)}_{r_i(t)} & \bar{\gamma_i}\phi^{(A)}_{r_{i}(t)}\bar{\phi}^{(S)}_{r_i(t)} & \gamma_i\bar{\phi}^{(A)}_{r_{i}(t)}\bar{\phi}^{(S)}_{r_i(t)} & \gamma_i\phi^{(A)}_{r_{i}(t)}\bar{\phi}^{(S)}_{r_i(t)} & \bar{\phi}^{(A)}_{r_{i}(t)}\phi^{(S)}_{r_i(t)} & \phi^{(A)}_{r_{i}(t)}\phi^{(S)}_{r_i(t)} & 0\\
0 & 0 & 0 & 0 & 0 & 0 & 1 \\
0 & 0 & 0 & 0 & 0 & 0 & 1 \\
0 & 0 & 0 & 0 & 0 & 0 & 1 \\
\end{pmatrix}.
\end{equation}

To further define the model in GCM, we rely on the activation matrices $\{I_{\tilde{p},t}^{(+)}, I_{\tilde{p},t}^{(-)}\}_{\tilde{p}=1,\hdots,\tilde{P}}^{t=t\hdots,T}$ and the activation functions $\{f_\text{act}^{p}\}_{p=1,\hdots,P}$. The activation matrices can be directly obtained from the transition matrix, using Equations \eqref{eq:paramactfunction}, \eqref{eq:paramactmat}, and \eqref{eq:czmtransmat}. To define the activation functions, we rely on what data is available. To somewhat simplify the model, we assume $\gamma_i=\gamma$ for all $i\in\mathcal{I}$, and we represent items by a two-dimensional location vector. Although we could optimize $\gamma$ using Expression \eqref{eq:onecovcase}, we will for uniformity numerically optimize the simple neural network $f_{\text{act}}^{\gamma}(\theta_{\gamma};1)=\sigma(\theta_{\gamma})$ w.r.t. $\theta_{\gamma}$, $\sigma(\cdot)$ being the sigmoid function.

\paragraph{UBM to GCM}
As discussed in Proposition \ref{prop:ubmisgcm}, we can model UBM as a GCM by passing information about the last clicked item through positions $1,\hdots, T$. The state of the system at some session/position $(i,t)$ is defined by $\bm{\psi}_t^{(i)}=(\mathbf{e}_t^{(i)}, R_{r_i(t)}^{(i)}, E_t^{(i)})$. For simplicity, we will replace $\mathbf{e}_t$ by $t^{'}$, i.e., the last clicked item before $t$, such that we can use the following state space mapping (Table \ref{tab:ubmspmapping}). Note that, although we have multiple click states, given $t^{'}$ there is still only one click state. 
\begin{table}[htbp]
    \centering
    \begin{tabular}{l|ccc|r}
    \toprule
        State & $t^{'}$ & $E_t^{(i)}$ & $R_{r_i(t)}^{(i)}$ & $y_t^{(i)}$ \\
        \hline
        $0$ & $0$ & $0$ & $0$ & $0$  \\  
        $1$ & $0$ & $1$ & $0$ & $0$  \\  
        $2$ & $0$ & $0$ & $1$ & $0$  \\  
        $3$ $(\mathcal{C}_t)$ & $1$ & $1$ & $1$ & $1$  \\  
        $4$ & $1$ & $0$ & $0$ & $0$  \\  
        $\hdots$ & $\hdots$ & $\hdots$ & $\hdots$ & $\hdots$ \\
        $4T+1$ & $T$ & $0$ & $0$ & $0$  \\  
        $4T+2$ & $T$ & $1$ & $0$ & $0$  \\  
        $4T+3$ & $T$ & $0$ & $1$ & $0$  \\  
        $4T+4$ $(\mathcal{C}_t)$ & $T$ & $1$ & $1$ & $1$  \\ 
        $O$       &  -  &  -  &  -  &  -\\ 
        \bottomrule
    \end{tabular}
    \caption{UBM state space mapping}
    \label{tab:ubmspmapping}
\end{table}

The UBM transition matrix is slightly more involved than that of CZM. First, we consider two $4\times 4$ matrices containing the probabilities of clicking/skipping some item at position $t$ given the last click was at position $t^{'}$. The click probability matrix is given by
\begin{equation}
M_{i,t^{'}t}^{(+)}=\mathbf{1}\begin{pmatrix}
0, & 0, & 0, & \phi^{(R)}\gamma_{t^{'}t}
\end{pmatrix},
\end{equation}
whereas the skip probability matrix is given by
\begin{equation}
M_{i,t^{'}t}^{(-)}=\mathbf{1}\begin{pmatrix}
\bar{\phi}^{(R)}\bar{\gamma}_{t^{'}t},& \bar{\phi}^{(R)}\gamma_{t^{'}t},& \phi^{(R)}\bar{\gamma}_{t^{'}t},& 0
\end{pmatrix}.
\end{equation}
Here $\bm{1}$ is a vector of ones with 4 elements. By combining these matrices into one transition matrix, we obtain
\begin{equation}
M_{0,t^{'}t}=\begin{pmatrix}
M_{i,0t}^{(-)} & M_{i,t^{'}1}^{+} & \hdots & M_{i,t^{'}T}^{+} & \mathbf{0}\\
[0]              & M_{i,1t}^{(-)}   & \hdots & M_{i,t^{'}T}^{+} & \mathbf{0}\\
\vdots         & \ddots            & \ddots   & \vdots          & \mathbf{0}\\
[0]              & \hdots           & M_{i,(T^{'}-1)T}^{-} & M_{i,(T^{'}-1)T}^{+} & \mathbf{0}\\
[0]              & \hdots           & \hdots &   & \mathbf{1}\\
0               & \hdots          & \hdots  &. \hdots  & 0
\end{pmatrix},
\end{equation}
with $[0]$ a $4\times 4$ matrix consisting of only zeros, and $[1]$ a $4 \times 4$ matrix consisting of only ones. The vectors $\mathbf{0}$ and $\mathbf{1}$ are both vectors of 4 elements with only zeros and ones respectively.

Again, the activation matrices can be found using Expressions  \eqref{eq:paramactfunction}, \eqref{eq:paramactmat}, and \eqref{eq:czmtransmat}. For the activation functions, we use the same strategy as with CZM: we present the items by their two-dimensional position vector. To model probabilities $\gamma_{t^{'}t}$ a vanilla activation function is used.

\paragraph{Comparison of perplexity}
After estimating the two models using the \verb|gecasmo| package on a simulated dataset, where the simulation is based on the method proposed by \citet{de2020detecting}, we obtain Figure \ref{fig:perplexityplot}. Overall, the perplexity is 1.38 and 1.77 for CZM and UBM respectively. In this particular example, UBM seems to be biased towards correctly predicting skips, and therefore scores badly for positions close to 1. 
\begin{figure}
\centering
\begin{tikzpicture}[x=1pt,y=1pt]
\definecolor{fillColor}{RGB}{255,255,255}
\path[use as bounding box,fill=fillColor,fill opacity=0.00] (0,0) rectangle (289.08,144.54);
\begin{scope}
\path[clip] (  0.00,  0.00) rectangle (289.08,144.54);
\definecolor{drawColor}{RGB}{255,255,255}
\definecolor{fillColor}{RGB}{255,255,255}

\path[draw=drawColor,line width= 0.6pt,line join=round,line cap=round,fill=fillColor] (  0.00, -0.00) rectangle (289.08,144.54);
\end{scope}
\begin{scope}
\path[clip] ( 30.93, 25.56) rectangle (228.88,139.04);
\definecolor{fillColor}{RGB}{255,255,255}

\path[fill=fillColor] ( 30.93, 25.56) rectangle (228.88,139.04);
\definecolor{fillColor}{gray}{0.20}

\path[fill=fillColor] ( 39.93, 30.72) rectangle ( 48.11, 50.47);

\path[fill=fillColor] ( 58.10, 30.72) rectangle ( 66.28, 49.57);

\path[fill=fillColor] ( 76.28, 30.72) rectangle ( 84.46, 47.39);

\path[fill=fillColor] ( 94.46, 30.72) rectangle (102.64, 45.51);

\path[fill=fillColor] (112.64, 30.72) rectangle (120.82, 44.05);

\path[fill=fillColor] (130.81, 30.72) rectangle (138.99, 43.03);

\path[fill=fillColor] (148.99, 30.72) rectangle (157.17, 42.31);

\path[fill=fillColor] (167.17, 30.72) rectangle (175.35, 41.83);

\path[fill=fillColor] (185.34, 30.72) rectangle (193.52, 41.45);

\path[fill=fillColor] (203.52, 30.72) rectangle (211.70, 41.19);
\definecolor{fillColor}{gray}{0.80}

\path[fill=fillColor] ( 48.11, 30.72) rectangle ( 56.29, 55.98);

\path[fill=fillColor] ( 66.28, 30.72) rectangle ( 74.46,133.88);

\path[fill=fillColor] ( 84.46, 30.72) rectangle ( 92.64, 51.65);

\path[fill=fillColor] (102.64, 30.72) rectangle (110.82, 50.32);

\path[fill=fillColor] (120.82, 30.72) rectangle (129.00, 44.20);

\path[fill=fillColor] (138.99, 30.72) rectangle (147.17, 42.87);

\path[fill=fillColor] (157.17, 30.72) rectangle (165.35, 42.38);

\path[fill=fillColor] (175.35, 30.72) rectangle (183.53, 41.82);

\path[fill=fillColor] (193.52, 30.72) rectangle (201.70, 41.88);

\path[fill=fillColor] (211.70, 30.72) rectangle (219.88, 41.20);
\end{scope}
\begin{scope}
\path[clip] (  0.00,  0.00) rectangle (289.08,144.54);
\definecolor{drawColor}{RGB}{0,0,0}

\path[draw=drawColor,line width= 0.6pt,line join=round] ( 30.93, 25.56) --
	( 30.93,139.04);
\end{scope}
\begin{scope}
\path[clip] (  0.00,  0.00) rectangle (289.08,144.54);
\definecolor{drawColor}{gray}{0.30}

\node[text=drawColor,anchor=base east,inner sep=0pt, outer sep=0pt, scale=  0.60] at ( 25.98, 28.66) {0.0};

\node[text=drawColor,anchor=base east,inner sep=0pt, outer sep=0pt, scale=  0.60] at ( 25.98, 53.39) {2.5};

\node[text=drawColor,anchor=base east,inner sep=0pt, outer sep=0pt, scale=  0.60] at ( 25.98, 78.12) {5.0};

\node[text=drawColor,anchor=base east,inner sep=0pt, outer sep=0pt, scale=  0.60] at ( 25.98,102.86) {7.5};

\node[text=drawColor,anchor=base east,inner sep=0pt, outer sep=0pt, scale=  0.60] at ( 25.98,127.59) {10.0};
\end{scope}
\begin{scope}
\path[clip] (  0.00,  0.00) rectangle (289.08,144.54);
\definecolor{drawColor}{gray}{0.20}

\path[draw=drawColor,line width= 0.6pt,line join=round] ( 28.18, 30.72) --
	( 30.93, 30.72);

\path[draw=drawColor,line width= 0.6pt,line join=round] ( 28.18, 55.46) --
	( 30.93, 55.46);

\path[draw=drawColor,line width= 0.6pt,line join=round] ( 28.18, 80.19) --
	( 30.93, 80.19);

\path[draw=drawColor,line width= 0.6pt,line join=round] ( 28.18,104.92) --
	( 30.93,104.92);

\path[draw=drawColor,line width= 0.6pt,line join=round] ( 28.18,129.66) --
	( 30.93,129.66);
\end{scope}
\begin{scope}
\path[clip] (  0.00,  0.00) rectangle (289.08,144.54);
\definecolor{drawColor}{RGB}{0,0,0}

\path[draw=drawColor,line width= 0.6pt,line join=round] ( 30.93, 25.56) --
	(228.88, 25.56);
\end{scope}
\begin{scope}
\path[clip] (  0.00,  0.00) rectangle (289.08,144.54);
\definecolor{drawColor}{gray}{0.20}

\path[draw=drawColor,line width= 0.6pt,line join=round] ( 48.11, 22.81) --
	( 48.11, 25.56);

\path[draw=drawColor,line width= 0.6pt,line join=round] ( 66.28, 22.81) --
	( 66.28, 25.56);

\path[draw=drawColor,line width= 0.6pt,line join=round] ( 84.46, 22.81) --
	( 84.46, 25.56);

\path[draw=drawColor,line width= 0.6pt,line join=round] (102.64, 22.81) --
	(102.64, 25.56);

\path[draw=drawColor,line width= 0.6pt,line join=round] (120.82, 22.81) --
	(120.82, 25.56);

\path[draw=drawColor,line width= 0.6pt,line join=round] (138.99, 22.81) --
	(138.99, 25.56);

\path[draw=drawColor,line width= 0.6pt,line join=round] (157.17, 22.81) --
	(157.17, 25.56);

\path[draw=drawColor,line width= 0.6pt,line join=round] (175.35, 22.81) --
	(175.35, 25.56);

\path[draw=drawColor,line width= 0.6pt,line join=round] (193.52, 22.81) --
	(193.52, 25.56);

\path[draw=drawColor,line width= 0.6pt,line join=round] (211.70, 22.81) --
	(211.70, 25.56);
\end{scope}
\begin{scope}
\path[clip] (  0.00,  0.00) rectangle (289.08,144.54);
\definecolor{drawColor}{gray}{0.30}

\node[text=drawColor,anchor=base,inner sep=0pt, outer sep=0pt, scale=  0.60] at ( 48.11, 16.48) {1};

\node[text=drawColor,anchor=base,inner sep=0pt, outer sep=0pt, scale=  0.60] at ( 66.28, 16.48) {2};

\node[text=drawColor,anchor=base,inner sep=0pt, outer sep=0pt, scale=  0.60] at ( 84.46, 16.48) {3};

\node[text=drawColor,anchor=base,inner sep=0pt, outer sep=0pt, scale=  0.60] at (102.64, 16.48) {4};

\node[text=drawColor,anchor=base,inner sep=0pt, outer sep=0pt, scale=  0.60] at (120.82, 16.48) {5};

\node[text=drawColor,anchor=base,inner sep=0pt, outer sep=0pt, scale=  0.60] at (138.99, 16.48) {6};

\node[text=drawColor,anchor=base,inner sep=0pt, outer sep=0pt, scale=  0.60] at (157.17, 16.48) {7};

\node[text=drawColor,anchor=base,inner sep=0pt, outer sep=0pt, scale=  0.60] at (175.35, 16.48) {8};

\node[text=drawColor,anchor=base,inner sep=0pt, outer sep=0pt, scale=  0.60] at (193.52, 16.48) {9};

\node[text=drawColor,anchor=base,inner sep=0pt, outer sep=0pt, scale=  0.60] at (211.70, 16.48) {10};
\end{scope}
\begin{scope}
\path[clip] (  0.00,  0.00) rectangle (289.08,144.54);
\definecolor{drawColor}{RGB}{0,0,0}

\node[text=drawColor,anchor=base,inner sep=0pt, outer sep=0pt, scale=  0.80] at (129.90,  7.06) {Item position};
\end{scope}
\begin{scope}
\path[clip] (  0.00,  0.00) rectangle (289.08,144.54);
\definecolor{drawColor}{RGB}{0,0,0}

\node[text=drawColor,rotate= 90.00,anchor=base,inner sep=0pt, outer sep=0pt, scale=  0.80] at ( 11.01, 82.30) {Perplexity};
\end{scope}
\begin{scope}
\path[clip] (  0.00,  0.00) rectangle (289.08,144.54);
\definecolor{fillColor}{RGB}{255,255,255}

\path[fill=fillColor] (239.88, 56.82) rectangle (283.58,107.79);
\end{scope}
\begin{scope}
\path[clip] (  0.00,  0.00) rectangle (289.08,144.54);
\definecolor{drawColor}{RGB}{0,0,0}

\node[text=drawColor,anchor=base west,inner sep=0pt, outer sep=0pt, scale=  0.80] at (245.38, 96.00) {Model};
\end{scope}
\begin{scope}
\path[clip] (  0.00,  0.00) rectangle (289.08,144.54);
\definecolor{fillColor}{gray}{0.20}

\path[fill=fillColor] (246.09, 77.48) rectangle (259.12, 90.51);
\end{scope}
\begin{scope}
\path[clip] (  0.00,  0.00) rectangle (289.08,144.54);
\definecolor{fillColor}{gray}{0.80}

\path[fill=fillColor] (246.09, 63.03) rectangle (259.12, 76.06);
\end{scope}
\begin{scope}
\path[clip] (  0.00,  0.00) rectangle (289.08,144.54);
\definecolor{drawColor}{RGB}{0,0,0}

\node[text=drawColor,anchor=base west,inner sep=0pt, outer sep=0pt, scale=  0.60] at (263.83, 81.93) {CZM};
\end{scope}
\begin{scope}
\path[clip] (  0.00,  0.00) rectangle (289.08,144.54);
\definecolor{drawColor}{RGB}{0,0,0}

\node[text=drawColor,anchor=base west,inner sep=0pt, outer sep=0pt, scale=  0.60] at (263.83, 67.48) {UBM};
\end{scope}
\end{tikzpicture}
\caption{Perplexity UBM vs CZM}
\label{fig:perplexityplot}
\end{figure}
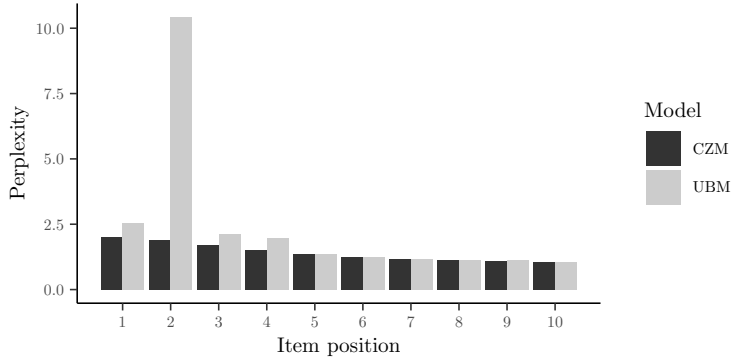

\section{Conclusion}
\label{sec:conclusion}
In this paper, we presented two contributions. First, we proposed an alternative view on estimating click models. We defined the generalized cascade model (GCM), and showed that the expectation maximization (EM) algorithm for GCMs is equivalent to the EM algorithm of an Input-Output Hidden Markov Model (IO-HMM), given that we add the observed click at position $t-1$ to the covariate vector at time $t$. As a consequence, we can directly use the EM algorithm for IO-HMMs, which provides us with an estimation procedure for GCMs without having to derive further expressions. Since many click models, including the User Browser Model (UBM) and Chapelle-Zang Model (CZM), can be written as GCM, this implies that we obtain an estimation procedure for many click models, making derivations of the E-step for specific click models (e.g., as is done in \cite{chuklin2015click}) obsolete.  

Second, we introduced the \verb|gecasmo| package, which not only includes an EM implementation for IO-HMMs, but also takes care of the mapping from GCM to IO-HMM. As a result, one only needs to define, for some specific GCM, some activation matrices and activation functions, from which the IO-HMM's transition matrix can be determined. To provide some more clarity on how to define such activation matrices and functions, we included an example for UBM and CZM. Also implementation-wise the usage of IO-HMM's EM algorithm is convenient: where previous packages for estimating click models relied on the user having to implement the EM updates manually, \verb|gecasmo| only requires the activation matrices and activation functions. Furthermore, as users can define their activation functions as neural networks, the model contains substantial flexibility in modeling click models.

Further research may move into several directions. First, in this paper we mainly focused on UBM and CZM  as examples, as these two click models are both generalization of many other click models, or are used as a starting point for more complex models. As we have seen in the UBM case, even models with a non-first-order examination assumption can be modeled as GCM by expanding the state space, and under small list sizes, do not suffer too much from the curse of dimensionality. We expect that even further generalizations to UBM or CZM can be mapped to a GCM. It would be interesting to consider to what extent these generalizations can be modeled as GCM, and if such mapping remains practical.

Second, we limited our scope to EM procedures, which does not imply that alternative estimation procedures, such as variational inference or Markov Chain Monte Carlo, could not benefit from viewing click models as GCM. On the contrary: it would be interesting to consider whether modeling click models as GCM, or more general as IO-HMM, is also convenient under these estimation techniques.

Third, as discussed in Section \ref{sec:relwork}, even though the usage of recurrent neural networks have shown to outperform many probabilistic graphical models (PGM) in terms of model perplexity, these models provide limited understanding of how users use search engines, and such understanding is obtained post hoc. As a result, user behavior may be misinterpreted. PGMs, on the other hand, require the analyst to define the model beforehand, which allows the analyst to test hypotheses on browsing behavior. What remains unknown, is to what extent using neural networks to map contextual data to a probability over a latent variable closes the gap in accuracy between the two types of models, without losing the interpretability of the model.



\bibliographystyle{ACM-Reference-Format}
\bibliography{papercontents/references}

\appendix
\section{Overview of notation}

\begin{table}[H]
    \centering
    \begin{tabular}{ll}
    \toprule
    \textbf{Variable} & \textbf{Interpretation}\\
    \hline
    $\{1,\hdots,n\}$ & The set of query sessions, indexed by $i$.\\
    $\{1,\hdots,V\}$ & Set of all items, indexed by $v$.\\
    $\mathcal{S}_i$ & Set of items in SERP corresponding to query session $i$.\\
    $\{1,\hdots,T\}$ & Set of all list positions, indexed by $t$.\\
    $\Omega$ & Set of all parameters in the GCM.\\
    $K$ & Size of the (possibly augmented) state space.\\
     $r_i(t)$ & The item in list position $v$ of query session $i$.\\
     $\mathcal{C}_t$ & The click state at time $t$\\
    $\bm{\psi}_{t}^{(i)}=(\psi_{t,1}^{(i)},\hdots,\psi_{t,P}^{(i)})$ & Vector of all latent variables at position $t$ of query session $i$, \\
    & which determines the latent state of GCM, \\
    & and can uniquely be mapped to $z_t^{(i)}$.\\
    $\mathbf{x}_{t,p}^{(i)}=(x_{t,p,1}^{(i)},\hdots,x_{t,p,m_p}^{(i)})$ & Vector of covariate values for parameter $p$.\\
    $\vartheta_{t,p}^{(i)}$ & Probability of a latent variable being one, given the previous state \\
    $f_{\text{act}}^{p}(\cdot)$ & Activation function for paramter $p$\\
    $\bm{\theta}_{p}$ & Weight vector for parameter $p$\\
    $y_{t}^{(i)}$ & 1 if the item in position $t$ of SERP $i$ is clicked, 0 otherwise.\\
    $\mathcal{P}(X)$ & The set of all parents of node $X$ in a dynamic Bayesian network.\\
    $z_t^{(i)}$ & The single-dimensional discrete state, \\
    & which uniquely corresponds to some $\bm{\psi}_{t}^{(i)}$.\\
    $\varphi_{k^{'}k,t}^{(i)}$, $M_{i,t}$ & Transition probability and transition matrix. \\
    $\zeta_{k,t}^{(i)}$ & Probability of being in state $k$ at $(i,t)$.\\
    $f_y(\cdot)$ & Emission probability of the IO-HMM.\\ 
    $h_{k^{'}k,t}$, $H_{i,t}$ & Expectation of the joint distribution of being in state $k$ at time $t$,\\
    & and state $k^{'}$ at time $t-1$, given all available data and \\
    & current estimate $\hat{\Omega}$.\\
    $c_{k^{'}k,t,p}^{+}$, $c_{k^{'}k,t,p}^{-}$, $I_{p,t}^{+}$, $I_{p,t}^{-}$ & Indicators of whether parameter $p$ positively ($\psi_p=1$) or \\ 
    &  negatively ($\psi_p=0$) influences the transition probability\\
    & between state $k^{'}$ and $k$.\\
    $\alpha_{k,t}^{(i)}$, $A_i$ & Forward probabilities in the forward-backward equations.\\
    $\beta_{k,t}^{(i)}$, $B_i$ & Backward probabilities in the forward-backward equations.\\
    $\delta_{kt}^{(i)}$, $D_i$ & Indicator of whether state $k$ is a click state\\
    $R_{v}^{(i)}$, $\phi_{v,i}^{(R)}$ & Indicator (probability) whether the user (of query session $i$)\\
    &  finds item $v$ attractive (also referred to as `relevance' in\\
    &  the click literature).\\
    $S_v^{(i)}$, $\phi_{v,i}^{(S)}$  & Indicator whether the user (of query session $i$)\\
    & is satisfied with item $v$ after a click.\\
    $E_{i,t}$ & Indicator whether the item at position $t$ is evaluated.\\
    $\gamma_i$ & Search continuation probability, has different sub/superscripts\\
    &  based on the underlying click model.\\
    \bottomrule
    \end{tabular}
    \caption{Notation overview}
    \label{tab:notationoverview}
\end{table}

\section{Notes on vector notation}
In the following we will elaborate shortly on the vectorization of Section \ref{subsubsec:implnotates}, in particular the expressions in Algorithm \ref{alg:estep}. From Bishop \cite[pp. 618-625]{bishop2006pattern} we find (using our own notation and ignoring superscript $i$)
\begin{equation}
\alpha_{k,t}=\P(y_t|z_t=k)\sum_{k^{'}=1}^{K}\varphi_{k^{'}k}(\mathbf{x}_t)\alpha_{k^{'},t-1},
\end{equation}
\begin{equation}
    \beta_{k,t}=\sum_{k^{'}=1}\P(y_{t+1}|z_{t+1}=k^{'})\varphi_{kk^{'}}(\mathbf{x}_{t+1})\beta_{k^{'},t+1},
\end{equation}
\begin{equation}
\label{eq:vecnot:hiddenexpr}
h_{k^{'}k,t}=\frac{\alpha_{k^{'},t-1}f_y(y_t|z_t=k)\varphi_{k^{'}k}(\mathbf{x}_t)\beta_{k,t}}{\sum_{l=1}^{K}\alpha_{l,T}}.
\end{equation}

Writing out $\alpha_{k,t}$ for $t>0$ gives
\begin{equation}
\begin{aligned}
\alpha_{k,t} &=\P(y_t|z_t=k)\sum_{k^{'}=1}^{K}\varphi_{k^{'}k}(\mathbf{x}_t)\alpha_{k^{'},t-1}\\
&=y_t\delta_{k,t}\sum_{k^{'}=1}^{K}\varphi_{k^{'}k}(\mathbf{x}_t)\alpha_{k^{'},t-1} + (1-y_t)(1-\delta_{k,t})\sum_{k^{'}=1}^{K}\varphi_{k^{'}k}(\mathbf{x}_t)\alpha_{k^{'},t-1},
\end{aligned}
\end{equation}
which in vector notation becomes
\begin{equation}
A_{1:K,t}=y_{t}\left[ D_{1:K,t} \odot \left(M_t^{\top} A_{1:K,t-1}\right)\right] + \left(1-y_{t}\right)\left[\left(1-D_{1:K,t}\right)\odot \left(M_{t} A_{1:K,t-1}\right)\right],
\end{equation}
where $\odot$ is the element-wise product. For $t=0$, using the original definition, we obtain:
\begin{equation}
\alpha_{k,0}=\P(y_0,z_0=t)=\delta_{0k},
\end{equation}
as we assume the item at position $t=0$ is clicked.

Similar, for $\beta_{k,t}$, $t<T$
\begin{equation}
\begin{aligned}
\beta_{k,t}&=\sum_{k^{'}=1}\P(y_{t+1}|z_{t+1}=k^{'})\varphi_{k^{'}k}(\mathbf{x}_{t+1})\beta_{k^{'},t+1}\\
&=\sum_{k^{'}=1}^{K}\varphi_{k^{'}k}(\mathbf{x}_{t+1})\beta_{k^{'},t+1}y_{t+1}f_y(1|z_{t+1=k})+\varphi_{k^{'}k}(\mathbf{x}_{t+1})\beta_{k^{'},t+1}(1-y_{t+1})f_y(0|z_{t+1=k}),
\end{aligned}
\end{equation}
which, again in vector notation, becomes
\begin{equation}
\begin{aligned}
B_{1:K,t}&=M_{t+1}^{\top}\left[y_{t+1}\left(B_{1:K,t+1}\odot D_{1:K,t+1}\right)+(1-y_{t+1})\left(B_{1:K,t+1}\odot \left(1-D_{1:K,t+1}\right)\right) \right]\\
&=M_{t+1}^{\top}\Lambda_{t+1},
\end{aligned}
\end{equation}
with
\begin{equation}
\Lambda_{t} =  y_{t}\left(B_{1:K,t}\odot D_{1:K,t}\right)+(1-y_{t})\left(B_{1:K,t}\odot \left(1-D_{1:K,t}\right)\right).
\end{equation}
We take $\beta_{k,T}=1$ for all $k=1,\hdots,K$.

Last, using \eqref{eq:vecnot:hiddenexpr}, and the previously found expressions, we obtain:
\begin{equation}
H_{t}= \left[\frac{\Lambda_{t}}{\sum_{k=1}^{K}\alpha_{k,T}} A_{1:K, t-1}^{\top}\right] \odot M_t^{\top}
\end{equation}

\section{Notes on click simulation}
Simulation parameters used in Section \ref{sec:package}. For brevity we only consider simulation parameters which were not obtained from the literature in \cite{de2020detecting}. Since the simulation is described in more depth in \cite{de2020detecting}, we will only provide a brief explanation of each variable here.
\begin{itemize}
    \item \textbf{Items}, number of items $V$.
    \item \textbf{Users}, total number of users, who may have multiple query sessions.
    \item \textbf{Warm-up sessions}, number of sessions used to estimate the overall item popularity. During the simulation, items are ordered at random, with the item order proportional to this overall item popularity.
    \item \textbf{List size}, number of items in a SERP, which we denoted by $T$.
    \item \textbf{User distance sensitivity}, a parameter regulating user attraction to nearby items.
    \item \textbf{Attraction salience}, parameter to regulate the attraction salience.
    \item \textbf{Satisfaction salience}, parameter to regulate the satisfaction salience.
    \item \textbf{Model lifetime phases geometric parameter}, the number of query sessions a single user generates is regulated by a sampled total user lifetime, and sampled times between two sessions (which in \cite{de2020detecting} is called \textit{absence time}). I.e., the number of query sessions one user generates is determined by whether the sum of absence times fits the user life time. This geometric parameters regulates the sampled user life time.
\end{itemize}
\begin{table}[htbp]
    \centering
    \begin{tabular}{ll}
    \toprule
    \textbf{Parameter} & \textbf{Value}\\
    \hline
    Items & 100\\
    Users & 20000\\
    Warm-up sessions & 100\\
    List size & 10 \\
    User distance sensitivity & 1\\
    Attraction salience & 5\\
    Satisfaction salience & 5\\
    Model lifetime phases geometric parameter  & 0.5\\
    Continuation probability & 0.9\\
        \bottomrule
    \end{tabular}
    \caption{Simulation parameters}
    \label{tab:simulationparams}
\end{table}

\end{document}